\documentclass[11pt]{article}
\usepackage[margin=1in]{geometry}
\usepackage{amsmath,amssymb,amsthm,mathtools}
\usepackage{booktabs,needspace}
\usepackage[hidelinks]{hyperref}
\usepackage{microtype}
\usepackage{xcolor}
\usepackage{amsmath,amssymb}
\usepackage{algorithm}
\usepackage{algpseudocode}

\hypersetup{pdftitle={Isomorphism Decision and Basis Construction
in Finite Abelian Groups},pdfauthor={Nader H. Bshouty}}

\newtheorem{theorem}{Theorem}[section]
\newtheorem{lemma}[theorem]{Lemma}
\newtheorem{corollary}[theorem]{Corollary}
\newtheorem{proposition}[theorem]{Proposition}
\theoremstyle{definition}
\newtheorem{example}[theorem]{Example}

\newcommand{\F}{\mathbb F}
\newcommand{\Z}{\mathbb Z}
\newcommand{\Prb}{\mathop{\rm Pr}}
\newcommand{\rank}{\mathop{\rm rank}}
\newcommand{\Bad}{\mathsf{Bad}}
\newcommand{\eval}{\mathsf{ev}}

\title{Almost Tight Bounds for Isomorphism Testing\\
and Basis Construction in Finite Abelian Groups}
\author{Nader H. Bshouty\\
Department of Computer Science, Technion, Haifa, Israel}
\date{}
\begin{document}
\maketitle

\begin{abstract}
We study isomorphism decision and basis construction for finite Abelian groups of known order \(n\). We measure complexity by the number of additions performed in the groups and by the total running time. 
We give a randomized isomorphism decision algorithm that performs \(\widetilde O(n^{1/4})\) group additions and runs in \(\widetilde O(n^{1/4})\) time. 
This improves the \(\widetilde O(\sqrt n)\) upper bound of Chen and Fu~\cite{ChenFu11} and matches Bshouty~\cite{Bshouty25} \(\Omega(n^{1/4})\) lower bound up to polylogarithmic factors. We also prove that finding a basis with success probability at least \(2/3\) requires \(\Omega(\sqrt n)\) group additions in the worst case. This matches the \(\widetilde O(\sqrt n)\) upper bound of Chen and Fu~\cite{ChenFu11} up to polylogarithmic factors.

These results separate isomorphism decision from basis construction in finite Abelian groups: their optimal worst-case complexities are \(\widetilde\Theta(n^{1/4})\) and \(\widetilde\Theta(\sqrt n)\), respectively.
\color{black}
\end{abstract}

\section{Introduction}
We study three fundamental computational problems for finite Abelian groups.
The first is the \emph{isomorphism problem}: given two Abelian groups of the
same order, decide whether they are isomorphic.
The second is the problem of finding the \emph{isomorphism type} of an
Abelian group, that is, representing it as a direct product of cyclic groups.
The third is the problem of finding an \emph{explicit basis}: in addition to
determining the cyclic factors, we find elements of the input group that
generate these factors, so that every group element has a unique
representation in terms of the basis.

Although these problems are closely related algebraically, their
computational complexities can be quite different. For example, knowing the
isomorphism type of a group does not necessarily tell us how to find an
explicit basis for it.

In this paper, we determine tight bounds for these problems. In particular,
we give algorithms achieving these bounds and prove matching lower bounds.
The precise definitions of the problems and our computational model are
given in Section~\ref{sec:background}.

\subsection{Previous work and the contribution of this paper}

We consider finite Abelian groups of known order $n$, with access to
their elements and group operation. We measure complexity by the number
of group additions and the total running time. A basis, together with
the orders of its elements, determines the group as a direct product
of cyclic groups. Converting this decomposition to invariant factors
therefore allows us to decide isomorphism. Thus, algorithms for basis
construction also give algorithms for isomorphism decision.

We first review the deterministic algorithms. Vikas~\cite{Vikas96}
gave an $O(n\log n)$-time algorithm for Abelian group isomorphism
and an $O(n)$-time algorithm for Abelian groups of prime-power order.
Kavitha~\cite{Kavitha07} subsequently obtained an $O(n)$-time algorithm
for general Abelian groups. For basis construction, Chen and
Fu~\cite{ChenFu11} gave a deterministic linear-time algorithm.
Karagiorgos and Poulakis~\cite{KaragiorgosPoulakis11} also obtained
linear-time basis algorithms when the prime factorization of $n$
is supplied. Bshouty~\cite{Bshouty25} gave a deterministic construction
of generators and triangular relations using $O(n)$ group additions
and $\widetilde O(n)$ time. These relations yield a basis and the
isomorphism type within the same bounds. Bshouty~\cite{Bshouty25} also proved $\Omega(n)$ lower bounds
for deterministic algorithms for both isomorphism decision and
basis construction, showing that the known linear-time algorithms
are optimal.

Sublinear algorithms for basis construction also have a long history.
Buchmann and Schmidt~\cite{BuchmannSchmidt05} gave an algorithm that,
from a supplied generating set $M$, computes the group structure using
$O(|M|\sqrt{|G|})$ group operations. For a generating set of logarithmic
size, this gives a $\widetilde O(\sqrt n)$ bound. Chen and
Fu~\cite{ChenFu11} developed randomized sublinear-time basis algorithms.
In particular, when the elements can be accessed through a list and
the prime factorization of $n$ is supplied, their algorithm samples
elements, obtains generating sets for the Sylow subgroups with high
probability, and applies the Buchmann--Schmidt algorithm. Its worst-case
running time is $\widetilde O(\sqrt n)$. Later,
Bshouty~\cite{Bshouty25} gave another
$\widetilde O(\sqrt n)$-time randomized algorithm using a different
approach: it constructs generators with triangular relations and
then obtains a basis by computing a Smith normal form. These basis
algorithms also solve isomorphism decision.

For lower bounds, Bshouty~\cite{Bshouty25} proved that randomized
algorithms require $\Omega(n^{1/4})$ operations for both isomorphism
decision and basis construction.
This left a gap between the fourth-root lower bound and the
square-root upper bound for each problem.

This paper closes both gaps up to factors polynomial in $\log n$.
We give a randomized algorithm that determines the isomorphism type
using $\widetilde O(n^{1/4})$ group additions and time, matching
Bshouty's lower bound up to these factors. For explicit basis
construction, we prove an $\Omega(\sqrt n)$ lower bound, matching
the existing $\widetilde O(\sqrt n)$ upper bounds up to these factors.
Thus, the two problems have different optimal worst-case exponents:
$1/4$ for isomorphism decision and $1/2$ for basis construction.

Our upper bound uses earlier algorithms whose complexity depends
on the structure of the group. Teske~\cite{Teske99} generalized the
Pohlig--Hellman discrete-logarithm method to finite Abelian groups.
Sutherland~\cite{Sutherland11} refined this approach and developed
basis algorithms with more precise bounds. For an Abelian $p$-group
of rank\footnote{The rank of a finite Abelian $p$-group is the
smallest number of elements needed to generate it.
Equivalently, if
$K\cong\Z_{p^{e_1}}\times\cdots\times\Z_{p^{e_s}}$
with all $e_i\geq1$, then $K$ has rank $s$.
The trivial group has rank zero.} $s$, the principal exponential term is $p^{s/2}$.
Here the rank is the number of nontrivial cyclic factors.
We state the precise consequence we use in
Proposition~\ref{prop:sutherland}. For each Sylow $p$-subgroup $P$,
we apply this routine to
$pP=\{px:x\in P\}$. The rank of $pP$ is at most
$\frac12\log_p|P|$, giving a fourth-root bound.
We then recover the isomorphism type of $P$ from that of $pP$
and the known order $|P|$.

Our lower-bound proof is related to the collision arguments used
by Shoup~\cite{Shoup97} for generic discrete logarithms.
The common idea is that an unexpected equality between computed
expressions reveals information about the group.
However, a discrete-logarithm lower bound does not directly imply
a basis-construction lower bound, since the required outputs differ.
We therefore analyze basis construction directly on a randomly
labelled elementary Abelian $2$-group. We compare the algorithm's
execution with a computation on formal sums and bound both the
probability of an unexpected equality and the probability that
the resulting output is a basis. This yields the square-root
lower bound. The same symbolic comparison also gives a proof of
Bshouty's fourth-root decision lower bound~\cite{Bshouty25}, which we include in
Section~\ref{sec:consequences}.
\color{black}

\subsection{Main results}

Write $\Z_m$ for the additive cyclic group of integers modulo $m$ and
$\F_p$ for the field of integers modulo a prime $p$. Thus $\F_2^r$ is
the group of binary vectors of length $r$, with coordinatewise addition
modulo two. Groups isomorphic to $\F_p^r$ are called
\emph{elementary Abelian $p$-groups}.

Throughout, the group order $n$ is known. Group elements are represented
by distinct labels, and each evaluation of the sum of two elements
counts as one group addition. We give upper bounds on both the number
of group additions and the total running time. Our lower bounds count
only group additions and therefore also imply lower bounds on
running time.

Our first result shows that constructing an explicit basis requires
$\Omega(\sqrt n)$ additions in the worst case, even when the
isomorphism type of the group is known in advance.

\begin{theorem}[Basis construction]\label{thm:main}
For $n=2^r$, every randomized algorithm that, with probability at least
$2/3$ on every labelling, outputs $r$ elements forming a basis of a
group isomorphic to $\F_2^r$ requires $\Omega(\sqrt n)$ group additions
in the worst case. This holds even when all element labels and the
identity label are supplied for free, and all computation other than
group additions is free.
\end{theorem}

\Needspace{7\baselineskip}
Our second result gives a faster algorithm for determining the
isomorphism type. The algorithm computes the orders of the cyclic
factors in a direct-product decomposition without constructing
a basis of the whole group.

\begin{theorem}[Isomorphism type and decision]\label{thm:iso}
The isomorphism type of a finite Abelian group of known order $n$
can be computed with probability at least $2/3$ using
$\widetilde O(n^{1/4})$ group additions and
$\widetilde O(n^{1/4})$ total time.
Consequently, two finite Abelian groups of known order can be tested
for isomorphism with probability at least $2/3$ within the same bounds,
where $n$ is their common order.
\end{theorem}

The isomorphism algorithm matches Bshouty's
$\Omega(n^{1/4})$ lower bound~\cite{Bshouty25} up to factors polynomial
in $\log n$. The basis lower bound matches the existing
$\widetilde O(\sqrt n)$ upper bounds, including those obtained from
the algorithms of Chen and Fu~\cite{ChenFu11},
Sutherland~\cite{Sutherland11}, and Bshouty~\cite{Bshouty25}.
Thus, these results determine the optimal worst-case exponents
for both problems and separate isomorphism decision from explicit
basis construction.

The following table summarizes the bounds. Each upper bound applies
to both the number of group additions and the total running time.
Each lower bound applies to the number of group additions alone.
The indicated families witness the lower bounds; hardness is not
asserted at every group order.

\begin{center}
\begin{tabular}{lll}
\toprule
Task & Upper bound & Lower bound and witnessing orders\\
\midrule
Isomorphism decision
  & $\widetilde O(n^{1/4})$
  & $\Omega(n^{1/4})$, $n=4^m$\\
Explicit basis construction
  & $\widetilde O(n^{1/2})$
  & $\Omega(n^{1/2})$, $n=2^r$\\
\bottomrule
\end{tabular}
\end{center}
\color{black}
\subsection{The two main ideas}

We briefly describe the two main ideas behind our lower and upper bounds.
The first explains why finding a basis requires $O(\sqrt{n})$ time,
while the second explains why determining the isomorphism type can be done
in $O(n^{1/4})$ time.

\paragraph{The lower-bound idea.}
The main difficulty in finding a basis is that the names of the group
elements reveal no information about their algebraic structure. To formalize
this intuition, we compare the actual execution of an algorithm with an
idealized symbolic execution.

In the symbolic execution, we treat every element computed by the algorithm
as a formal sum of unknown vectors. Whenever a new formal sum is obtained,
we assign it a fresh random label. The same label is reused only when exactly
the same formal sum appears again. Thus, as long as two distinct formal sums
do not represent the same group element, the labels seen by the algorithm
reveal no information about the hidden vectors.

In this idealized execution, the algorithm cannot choose a basis more
effectively than by choosing independent uniformly random vectors. Its
success probability is therefore at most the probability that $r$
independent uniformly random binary vectors form a basis.

It remains to compare the symbolic execution with the actual execution.
The two executions can differ only if two distinct formal sums evaluate to
the same group element. If the algorithm produces at most $M$ such
expressions, there are at most $\binom{M}{2}$ pairs of expressions.
For each pair, the probability of such a collision is $1/n$. Hence, by a
union bound, the probability that any collision occurs is at most
${\binom{M}{2}}/{n}$.
This quantity remains small when $M=o(\sqrt{n})$, which gives the
square-root lower bound.

\paragraph{The upper-bound idea.}
For the upper bound, consider a $p$-group $P$ of order $p^a$. Suppose its
cyclic decomposition contains factors of orders
$p^{e_1},\ldots,p^{e_r}$.
Multiplying every element of $P$ by $p$ gives the subgroup
$pP=\{px:x\in P\}$.
Under this operation, every cyclic factor of order $p$ disappears, while a
factor of order $p^e$, for $e\ge 2$, becomes a cyclic factor of order
$p^{e-1}$.

The important observation is that the rank of $pP$ is at most $a/2$.
Indeed, every cyclic factor that survives in $pP$ had order at least $p^2$
in $P$, and therefore contributed at least two powers of $p$ to
$|P|=p^a$. Consequently, there can be at most $a/2$ surviving factors.

We can therefore apply the rank-sensitive basis algorithm to $pP$. Since
$pP$ has rank at most $a/2$, its running time is, up to logarithmic factors,
at most
$p^{a/4}$.
Once the cyclic factors of $pP$ are known, the known order of $P$ determines
how many factors of order $p$ disappeared when passing from $P$ to $pP$.
Thus, to determine the isomorphism type of $P$, we do not need to find
generators for these missing order-$p$ factors; we only need to determine
how many there are. This is what allows the isomorphism-type problem to be
solved faster than the explicit-basis problem.

Section~\ref{sec:background} introduces the algebraic facts and computational
model used throughout the paper. Section~\ref{sec:coupling} formalizes the
symbolic execution described above, and Section~\ref{sec:basis-lower} uses
it to prove Theorem~\ref{thm:main}. Section~\ref{sec:upper} develops the
upper-bound argument and proves Theorem~\ref{thm:iso}.
Finally, Section~\ref{sec:consequences} gives the matching lower bound for
the decision problem and discusses the consequences for computing explicit
isomorphisms.

\section{Background and computational conventions}\label{sec:background}

\subsection{Cyclic decompositions, types, and bases}

We use additive notation: the identity is $0$, and $cx$ is the sum of $c$
copies of $x$ for a nonnegative integer $c$. The order of an element $x$, denoted by $\operatorname{ord}(x)$, is the
least positive $c$ for which $cx=0$. We write $\langle S\rangle$ for the
subgroup generated by a set or list $S$. A finite direct sum, denoted by
$\oplus$, is a Cartesian product with addition in each coordinate; it is
also called a direct product.

\Needspace{8\baselineskip}
Every finite Abelian group is a direct product of cyclic groups
of prime-power order. Two such groups are isomorphic if and only
if the lists of these prime-power orders, including repetitions,
agree after sorting. Our algorithm computes this list.
We recall the classification theorem below; see
Shoup~\cite[Section~6.6]{Shoup08} for a proof.

\begin{theorem}[Fundamental theorem of finite Abelian groups]\label{thm:structure}
Let $G$ be a finite Abelian group of order
$n=\prod_{p\mid n}p^{a_p}$, where $p$ ranges over the prime divisors of $n$.
For each such $p$, the elements whose orders are powers of $p$, together
with $0$, form a subgroup $G_p$ of order $p^{a_p}$. Moreover,
 $G\cong\bigoplus_{p\mid n}G_p$.

For each $p$ there is a unique nonincreasing list of positive integers
$\lambda_{p,1},\ldots,\lambda_{p,t_p}$ such that
$G_p\cong\bigoplus_{i=1}^{t_p}\Z_{p^{\lambda_{p,i}}}$.
The sum of the entries $\lambda_{p,i}$ in this list is $a_p$.
\end{theorem}

Theorem~\ref{thm:structure}  uniqueness says that two groups are
isomorphic exactly when these lists agree for every prime. We call the
collection of lists the \emph{type}. Each list is a partition of $a_p$,
meaning a nonincreasing list of positive integers summing to $a_p$.
The subgroup $G_p$ is called the \emph{$p$-primary component}, or the Sylow
$p$-subgroup; no further facts about Sylow theory will be needed.

A basis records elements realizing a cyclic decomposition. More precisely,
a tuple of nonidentity elements $(g_1,\ldots,g_t)$ is a \emph{basis} if every
$x\in G$ has a unique expression
\[
 x=c_1g_1+\cdots+c_tg_t,
 \qquad 0\leq c_i<\operatorname{ord}(g_i).
\]
For a $p$-group, the basis-element orders are exactly the orders of the cyclic groups in the decomposition given by
Theorem~\ref{thm:structure}. For a general group, bases can also
combine factors of relatively prime orders. Our algorithm records the
prime-power decomposition, and the lower bound uses only $2$-groups, so
this choice does not affect the results.

\begin{example}\label{ex:type}
The group $\Z_8\oplus\Z_4\oplus\Z_2\oplus\Z_2$ has order $2^7$ and type
$(3,2,1,1)$ at the prime $2$. Its four coordinate vectors form a basis.
If its $128$ elements are renamed by arbitrary labels, the type remains
$(3,2,1,1)$, but the labels of those coordinate vectors are hidden. Producing
the four integers and producing four suitable labels are different tasks.
\end{example}

\subsection{Multiplication by a prime and rank}\label{S22}

The upper bound needs a measure of how many cyclic factors remain after
multiplication by $p$. For an Abelian $p$-group $K$, define
\[
 pK=\{px:x\in K\},
 \qquad K[p]=\{x\in K:px=0\}.
\]
Both are subgroups: the first is the image of multiplication by $p$, and
the second is its kernel. The subgroup $K[p]$ consists of the identity
and the elements of order $p$.

The quotient $K/pK$ and the subgroup $K[p]$ are vector spaces over $\F_p$:
addition is their group operation, and scalar multiplication is repeated
addition modulo $p$. This is well defined because multiplication by $p$
is zero in each group. If $K$ has $s$ nontrivial\footnote{Not $\{0\}$.} cyclic factors, each factor
contributes one copy of $\Z_p$ to both $K/pK$ and $K[p]$. By
Theorem~\ref{thm:structure}, both vector spaces therefore have dimension
$s$. We call this number the \emph{rank} of $K$ and write $\rank(K)$.
The trivial group has rank zero.

For example, the group in Example~\ref{ex:type} has rank four because $K[2]=\{0,4\}\oplus \{0,2\}\oplus \Z_2^2$ which is isomorphic to $\Z_2^4$. Its image
under multiplication by two is isomorphic to $\Z_4\oplus\Z_2$, which has
rank two. The two factors of order two have disappeared entirely.

One useful consequence will also control the algorithm on unsuccessful
sampling runs. If $U$ is a subgroup of an Abelian $p$-group $K$, then
$U[p]$ is a vector subspace of $K[p]$. The dimension interpretation above
gives
\begin{equation}\label{eq:subgroup-rank}
 \rank(U)\leq\rank(K).
\end{equation}
Finally, an integer \emph{annihilates} a group if its multiple of every
element is zero. The \emph{exponent} is the least positive such integer.
An Abelian $p$-group has exponent $p^h$, where $h$ is the largest entry
in its type.

\subsection{Input, output, and complexity conventions}

The input supplies the group order $n$. Each group element has a distinct
label, and we identify the set of labels with
$L=[n]=\{1,\ldots,n\}$. Given two labels, we can perform a group addition
to obtain the label of their sum. Choosing a uniform label therefore
samples a uniform group element without performing any group additions.
Labels and indices use $O(\log n)$ bits. We measure both the number of
group additions and the total running time. For the upper bounds, each
group addition takes one unit of time, and ordinary computation is
charged in a random-access machine. Polynomial factors in $\log n$
arising from arithmetic and sorting are included in the
$\widetilde O$ notation.

For the lower bounds, we charge only group additions. Every label is
available for free, together with a distinguished identity label
$\iota$. All comparisons, random choices, and other computation are
free. The algorithm may choose the elements to be added using the
numerical values of their labels, including labels it has never used
before. Thus, the lower bounds apply even to algorithms that perform
arbitrary computations on labels and choose arbitrary group elements
as operands.

An explicit basis must be output as labels of group elements.
For a group isomorphic to $\F_2^r$, this means exactly $r$ labels.
If the algorithm first constructs formal expressions for the basis
elements, the group additions needed to evaluate those expressions
and obtain their labels must also be counted. The same requirement
applies if the output is a procedure for computing the basis elements.
If a basis convention permits identity elements, these can be removed
by comparison with the supplied identity label.
\color{black}

An algorithm performing at most $q$ group additions can be represented
by a decision tree. Each internal node specifies two labels to be
added, its outgoing edges correspond to the possible labels of their
sum, and its leaves specify outputs. Fixing the random choices makes
the tree deterministic. We define its behavior arbitrarily on answer
sequences that cannot arise from a genuine group and terminate every
execution after at most $q$ additions. This allows us to run the same
tree with the symbolic addition procedure defined below, even when
its answers do not describe a genuine group.
\section{An unconditional symbolic comparison}\label{sec:coupling}
Fix an Abelian group $H$ of order $n$. Assume that $dh=0$ for every
$h\in H$, where $d\geq 2$ is fixed. We use $d=2$ for the basis
lower bound and $d=4$ for the isomorphism lower bound.

The elements of $H$ are presented to the algorithm using labels from
$L=[n]$. To assign the labels, choose a bijection
$\phi:L\longrightarrow H$ uniformly at random, subject to
$\phi(\iota)=0$. Thus, $\phi(x)$ is the group element represented by
label $x$, and $\iota$ is the label of the identity.

The algorithm is given the labels but not the map $\phi$.
When it adds the elements labelled $x$ and $y$, it receives the label
\[
    x\star y
    =\phi^{-1}\bigl(\phi(x)+\phi(y)\bigr).
\]
In other words, the returned label represents the sum of the two
corresponding elements of $H$. The algorithm may use the numerical
values of the labels in its computations, but the assignment of
labels to elements is random.

To analyze what the algorithm can learn from group additions, we
simulate its execution using formal sums of symbols
$X_1,X_2,\ldots$, with coefficients taken modulo $d$.
We associate a formal sum with each label used in the simulation.
We assign $0$ to the identity label and, whenever the algorithm
chooses a label that has not yet been assigned a formal sum as
an operand or an output element, we assign to it the next unused
symbol $X_1,X_2,\ldots$.
When the algorithm requests the sum of two labelled elements, we
add their formal sums and return the label associated with the
result. The algorithm sees only the labels, not the formal sums.
The precise rules for assigning labels are given below.

For example, suppose labels $x$, $y$, and $z$ are associated with
$X_1$, $X_2$, and $X_1+X_2$, respectively. When $d=2$, adding the
elements labelled $z$ and $x$ returns label $y$, because
    $(X_1+X_2)+X_1=X_2$.
In general, two formal sums are considered equal if and only if
all their coefficients agree modulo $d$.

No elements of $H$ are used during this simulation. We first
complete the simulated execution, including all the algorithm's
addition requests, the returned labels, and the final output.
Afterwards, we replace each symbol $X_i$ by an independently
chosen uniform element $h_i\in H$. A formal sum
$\sum_i c_iX_i$ then represents the element $\sum_i c_i h_i$ of $H$.

Two distinct formal sums may represent the same element of $H$
after this substitution. We call this an \emph{unexpected equality}.
We bound the probability of such equalities using the independent
choices of the elements $h_i$. The formal sums themselves are
already fixed before these elements are chosen.

These independent choices do not define the random bijection
$\phi$: a bijection assigns distinct group elements to distinct
labels, whereas independent choices may repeat. They are used
only to analyze the simulation. We prove below that the simulated
execution and an execution with a uniformly random labelling of
$H$ can be arranged to agree unless an unexpected equality occurs.

\subsection{The symbolic simulation}

We now give the precise rules for the simulation described above.
A \emph{formal expression} is a finite sum
$\sum_j c_jX_j$, where the $X_j$ are formal symbols and the
coefficients belong to $\Z_d$. We add expressions by adding their
coefficients modulo $d$. Two expressions are equal if and only if
all their coefficients agree modulo $d$. These expressions do not
yet represent particular elements of $H$.

Fix the algorithm's random choices, so that its behavior is
deterministic given the answers to its addition requests.
We maintain a record $D$ that associates labels with formal
expressions. We call a label \emph{registered} once it has been
assigned an expression. Different registered labels always have
different expressions. Initially, only the identity label is
registered, with $D(\iota)=0$.

The simulation follows three rules.
\begin{enumerate}
\item
When the algorithm requests an addition, process the first operand
label and then the second. If an operand label has not been
registered, assign to it the next unused symbol
$X_1,X_2,\ldots$.

\item
Add the two operand expressions. If the resulting expression
already has a label in $D$, return that label. Otherwise, choose
a label uniformly among all unregistered labels, assign the
resulting expression to it, and return it.

\item
When the algorithm terminates, process its output labels in their
output order. Assign the next unused symbol to each output label
that has not yet been registered.
\end{enumerate}

Once a label has been assigned an expression, that assignment
never changes. The simulation records the requested additions,
the returned labels, the assignments in $D$, and the final output.
Since the algorithm's random choices are fixed, the only randomness
in this simulation is the choice of unused labels in the second rule.

\begin{example}
Take $d=2$ and suppose the algorithm requests the addition of two
distinct, previously unregistered labels $a$ and $b$.
We assign $D(a)=X_1$ and $D(b)=X_2$. Their addition returns a new
label $c$, with $D(c)=X_1+X_2$. If the algorithm subsequently adds
the elements labelled $c$ and $a$, the simulation returns $b$,
because
    $(X_1+X_2)+X_1=X_2$.
This answer follows from the rules for formal addition.
\end{example}

After completing the simulation, choose independent uniform elements
$h_1,h_2,\ldots$ of $H$, independently of all label choices made
during the simulation. Replace each formal symbol $X_j$ by $h_j$.
This defines the evaluation
\[
    \eval\!\left(\sum_j c_jX_j\right)=\sum_j c_jh_j.
\]
The assumption $dh=0$ for every $h\in H$ ensures that reducing
coefficients modulo $d$ does not change the evaluated group element.
Evaluation also preserves addition.

Let $\mathcal E$ be the set of distinct expressions registered
during the simulation, including $0$ and the expressions assigned
to output labels. Two different expressions in $\mathcal E$ may
evaluate to the same element of $H$. We call such an equality a
\emph{collision}, or an \emph{unexpected equality}, and write
$\Bad$ for the event that at least one occurs. Thus,
\[
    \Bad
    =
    \bigl\{
    \text{there exist }u,v\in\mathcal E,\ u\ne v,
    \text{ with }\eval(u)=\eval(v)
    \bigr\}.
\]
The equality in the example is not a collision: the expressions
$(X_1+X_2)+X_1$ and $X_2$ are already formally equal.

If the algorithm performs at most $q$ group additions and outputs
$k$ labels, then
\begin{equation}\label{eq:count}
    |\mathcal E|\leq M:=3q+k+1.
\end{equation}
Indeed, each addition registers at most two operand labels and
one result label. The output introduces at most $k$ further labels,
and the identity contributes one more.

The algorithm may also compute with or compare other labels.
These labels need no assignment until they are used as operands
or included in the output, since such local computations do not
use the group operation. We use the simulation when $M\leq n$,
which guarantees that an unused label is available whenever
one is needed.

\subsection{Comparing the actual and simulated executions}

Since $\phi:L\to H$ is a bijection, distinct labels represent
distinct group elements: if $x\ne y$, then $\phi(x)\ne\phi(y)$.  In contrast, the independent elements
$h_1,h_2,\ldots$ used to evaluate the formal symbols may repeat.
Thus, the actual execution and the symbolic simulation are not
identical experiments.

We will generate the two executions together, while preserving
the correct distribution of each one separately. Such a joint
construction is called a \emph{coupling}. Our construction keeps
the two executions in agreement unless a collision occurs among
the evaluated formal expressions.

\begin{lemma}[Coupling]\label{lem:coupling}
Suppose $M\leq n$. The symbolic simulation, together with its
independent random evaluations, can be coupled to an actual
execution on a uniformly random labelling of $H$ so that, whenever
$\Bad$ does not occur, the requested additions, returned labels,
and outputs are identical. Moreover, in this case,
\[
    \phi(\ell)=\eval(D(\ell))
\]
for every registered label $\ell$.
\end{lemma}

\begin{proof}
We construct the random bijection $\phi$ gradually, assigning a
group element to a label only when needed. Initially,
$\phi(\iota)=0$. At every stage, all bijections extending the
assignments already made must remain equally likely.

For this proof, we may also choose $h_j$ when the formal symbol
$X_j$ is first introduced. This gives the same symbolic experiment
as choosing all the $h_j$ after the simulation: each $h_j$ is chosen
independently and uniformly from $H$, and the symbolic simulation
never uses these values to choose its answers.

We try to keep both executions identical and to maintain
$\phi(\ell)=\eval(D(\ell))$ for every registered label.
We declare a divergence if these conditions can no longer be
maintained.

\paragraph{A new operand or output label.}
Suppose the executions still agree, and $\phi$ has been assigned
on $b$ labels. Let $\ell$ be a new label chosen by the algorithm
as an operand or included in its output. In a uniformly random
bijection, $\phi(\ell)$ must be uniform among the $n-b$ elements
of $H$ that have not yet been assigned to labels.

The simulation assigns the next unused symbol $X_j$ to $\ell$.
Choose its value $h_j$ independently and uniformly from all of $H$.
If $h_j$ has not yet been assigned to a label, set
$\phi(\ell)=h_j$. Otherwise, declare divergence and choose
$\phi(\ell)$ uniformly among the unused elements of $H$.

This rule gives the correct distribution for $\phi(\ell)$.
Indeed, for any unused element $a\in H$, conditional on the
preceding history, the probability that $\phi(\ell)=a$ is
\[
    \frac1n+\frac bn\frac1{n-b}
    =\frac1{n-b}.
\]
The first term corresponds to choosing $h_j=a$. The second
corresponds to choosing an already assigned element and then
choosing $a$ as the replacement.

If divergence occurs, $h_j$ equals the value of an earlier
registered expression. That expression is different from the
new symbol $X_j$, so this gives a collision in the symbolic
simulation. This argument also applies to output labels that
were not used earlier.

\paragraph{An addition answer.}
Suppose the executions still agree and the algorithm requests
the addition of labels $x$ and $y$. Let
    $w=D(x)+D(y)$.
Because evaluation preserves addition, the actual sum has value
$\phi(x)+\phi(y)=\eval(w)$.
If $w$ already has a registered label, both executions return
that label.
Otherwise, the simulation chooses a uniformly random unregistered
label $z$ and assigns $D(z)=w$. This choice is made without using
$\eval(w)$.

If $\eval(w)$ has not yet been assigned to a label in the actual
execution, set $\phi(z)=\eval(w)$ and return $z$ there as well.
This preserves the correct distribution of the random bijection:
the label of a specified unused group element is uniform among
the unregistered labels.

If $\eval(w)$ is already assigned to a label, the actual execution
returns that existing label, whereas the simulation returns the
new label $z$. Declare divergence. The expression associated
with the existing label is different from $w$, but has the same
evaluation. Thus, a collision has occurred.

\paragraph{Completing the executions.}
After a divergence, continue the actual execution by extending
its partially assigned bijection uniformly as needed. Continue
the symbolic simulation using its original rules, choosing new
answer labels without consulting the values $h_j$.

Each execution therefore retains its required distribution.
Every first divergence produces a collision in the symbolic
simulation. Consequently, if $\Bad$ does not occur, the executions
and outputs agree, and
$\phi(\ell)=\eval(D(\ell))$ holds for every registered label.
Complete $\phi$ uniformly on any remaining labels.

Throughout this construction, the symbolic experiment retains
its original distribution; we never condition it on the absence
of collisions.
\end{proof}

The lemma reduces the analysis of the actual algorithm to two
questions: how likely a collision is, and how likely the
evaluated output of the symbolic simulation is to form a basis.
We answer these questions next for binary vector spaces.
\color{black}

\section{The square-root basis lower bound}\label{sec:basis-lower}
Throughout this section, $V=\F_2^r$ and $n=2^r$.
Let $\mathcal D_r$ be the distribution on labelled copies of $V$
obtained by choosing a bijection $\phi:L\to V$ uniformly at random,
subject to $\phi(\iota)=0$.

We use the symbolic simulation of Section~\ref{sec:coupling},
with coefficients modulo two. For now, fix the algorithm's random
choices. Let $T$ record the complete simulated execution: the
requested additions, the returned labels, the assignments of
formal expressions to labels, and the output.

As before, $X_1,X_2,\ldots$ are formal symbols. After the simulation
is complete, we assign independent uniform vectors
$h_1,h_2,\ldots\in V$ to these symbols and evaluate the formal
expressions. All lower bounds below count group additions alone
and therefore also apply to the total running time.

\subsection{Collisions and the simulated output}

The algorithm may choose each addition based on the answers to
earlier additions. Nevertheless, the symbolic simulation does
not use the vectors $h_j$ to determine its answers. Thus, after
fixing $T$, all the formal expressions are fixed, while the
vectors $h_j$ remain independent and uniform in $V$.

\begin{lemma}[Unconditional collision bound]\label{lem:collision}
If the symbolic simulation registers at most $M$ distinct formal
expressions, then
\[
    \Prb[\Bad]\leq\frac{\binom M2}{n}.
\]
\end{lemma}

\begin{proof}
Fix the complete record $T$ of the symbolic simulation.
Let $u$ and $v$ be two distinct registered expressions.
Since their coefficient lists differ, we can write
    $u-v=\sum_j c_jX_j$
with at least one coefficient equal to one. Choose an index $t$
such that $c_t=1$.
If we fix all the vectors $h_j$ except $h_t$, then
\[
    \eval(u-v)=h_t+\sum_{j\ne t}c_jh_j.
\]
The second term is fixed, and $h_t$ is uniform in $V$. Adding a
fixed vector to a uniform vector gives another uniform vector.
Consequently,
$\Prb[\eval(u)=\eval(v)\mid T]=1/n$. There are at most $\binom M2$ pairs of distinct registered
expressions. By the union bound,
\[
    \Prb[\Bad\mid T]\leq\frac{\binom M2}{n}.
\]
This holds for every possible $T$. Averaging over $T$ proves
the stated bound. We have not conditioned on the absence of
collisions.
\end{proof}
We next examine the probability that the evaluated output is a
basis. First consider $r$ independent uniform vectors in $V$.
After $i$ linearly independent vectors have been chosen, their
span contains $2^i$ of the $2^r$ vectors in $V$. The next uniform
vector lies outside that span with probability $1-2^{i-r}$.
Therefore, the probability that all $r$ vectors form a basis is
\begin{equation}\label{eq:rho}
    \rho_r
    :=\prod_{i=0}^{r-1}(1-2^{i-r})
    =\prod_{j=1}^r(1-2^{-j}).
\end{equation}

The following lemma shows that choosing formal sums without
using their values in $V$ cannot increase this probability.

\begin{lemma}[Simulated output probability]\label{lem:ideal-output}
The evaluations of the $r$ output expressions in the symbolic
simulation form a basis of $V$ with probability at most $\rho_r$.
\end{lemma}

\begin{proof}
Fix $T$, and suppose the simulation introduced the symbols
$X_1,\ldots,X_t$. Write the output expressions as
\[
    u_i=\sum_{j=1}^t a_{ij}X_j,
    \qquad i=1,\ldots,r,
\]
and let $A=(a_{ij})$ be their $r\times t$ coefficient matrix
over $\F_2$. Since $T$ is fixed, $A$ is fixed independently of
the random vectors $h_1,\ldots,h_t$.
If the row rank of $A$ is less than $r$, some nonzero linear
combination of its rows is zero. The same combination of the
evaluated vectors $\eval(u_1),\ldots,\eval(u_r)$ is then zero.
These vectors are linearly dependent and cannot form a basis.

Suppose instead that $A$ has row rank $r$. Choose $r$ columns
forming an invertible $r\times r$ matrix. Fix the vectors $h_j$
corresponding to all the other columns.
For any prescribed output tuple $(y_1,\ldots,y_r)\in V^r$,
the equations
\[
    \sum_{j=1}^t a_{ij}h_j=y_i,
    \qquad i=1,\ldots,r,
\]
now have exactly one solution for the $r$ remaining vectors
$h_j$. Indeed, we can solve the invertible linear system
separately in each coordinate of $V$.
The $r$ remaining vectors are independent and uniform in $V$.
Thus, every output tuple in $V^r$ has probability $n^{-r}$.
The evaluated output vectors are therefore independent and
uniform in $V$, and their probability of forming a basis is
exactly $\rho_r$.

For every fixed $T$, the basis probability is consequently
either zero or $\rho_r$. Averaging over $T$ proves the lemma.
\end{proof}

Both lemmas use the original independent choices of the vectors
$h_j$. We must not first condition on the absence of collisions.
For example, two independent uniform vectors in $\F_2^2$ form
a basis with probability $3/8$. If we condition on their being
distinct and nonzero, they form a basis with probability one.
Thus, the conclusion of Lemma~\ref{lem:ideal-output} need not
hold after such conditioning.

\subsection{From the simulation to an addition lower bound}

By Lemma~\ref{lem:coupling}, whenever $\Bad$ does not occur,
the actual and simulated executions have the same output labels,
and those labels represent the evaluated output expressions.
Therefore, if the actual algorithm succeeds, either the evaluated
simulated output forms a basis or $\Bad$ occurs. Indeed, let $A$ be the event that the actual output is a basis,
let $S$ be the event that the evaluated simulated output is a
basis, and put $B=\Bad$. Since the two outputs represent the
same group elements whenever $B$ does not occur,
\[
\begin{aligned}
A
=(A\cap B^c)\cup(A\cap B)=(S\cap B^c)\cup(A\cap B)\subseteq S\cup B.
\end{aligned}
\]
Thus, actual success implies that either the evaluated simulated
output is a basis or $\Bad$ occurs. Taking probabilities gives
$\Prb[A]\leq \Prb[S]+\Prb[\Bad]$.
We bound these
two probabilities separately. 

Recall that $\mathcal D_r$ is the distribution on labelled copies of $V$
obtained by choosing a bijection $\phi:L\to V$ uniformly at random,
subject to $\phi(\iota)=0$.

\begin{theorem}[Quantitative success bound]\label{thm:quantitative}
Let an algorithm perform at most $q$ group additions and output
$r$ labels. If $M=3q+r+1\leq n=2^r$, then
\begin{equation}\label{eq:success}
    \Prb_{G\sim\mathcal D_r}
    [\text{the output is a basis}]
    \leq
    \rho_r+\frac{\binom{3q+r+1}{2}}{n}.
\end{equation}
For a randomized algorithm, the probability also includes its
internal random choices, and the bound of $q$ additions must
hold on every execution.
\end{theorem}

\begin{proof}
First fix the algorithm's random choices. Use the joint
construction from Lemma~\ref{lem:coupling}. 
By~\eqref{eq:count}, the simulation registers at most
$3q+r+1$ distinct expressions. Lemma~\ref{lem:ideal-output}
bounds the first probability by $\rho_r$, and
Lemma~\ref{lem:collision} bounds the second by
$\binom{3q+r+1}{2}/n$. This proves~\eqref{eq:success}
for every fixed choice of the algorithm's random bits.
Averaging over those choices proves the randomized statement.
\end{proof}

\begin{proof}[Proof of Theorem~\ref{thm:main}]
For $r\geq2$, equation~\eqref{eq:rho} gives
$\rho_r\leq\rho_2=3/8$.
An algorithm that succeeds with probability at least $2/3$
on every labelling also has average success probability at
least $2/3$ under $\mathcal D_r$.
Suppose first that $3q+r+1\leq n$. By
Theorem~\ref{thm:quantitative},
\[
    \frac{\binom{3q+r+1}{2}}{n}
    \geq \frac23-\frac38
    =\frac7{24}.
\]
Using $\binom M2\leq M^2/2$, we obtain
    $q
    =\Omega(\sqrt n)$.

In the remaining case, $3q+r+1>n$, so
    $q=\Omega(n)$,
which is stronger. Thus, the algorithm requires
$\Omega(\sqrt n)$ group additions in the worst case.
\end{proof}

\color{black}

\section{A fourth-root algorithm for isomorphism type}\label{sec:upper}
Computing the isomorphism type requires only the orders
of the cyclic groups in a direct-product decomposition. For each
prime $p$, we recover this information for the $p$-primary subgroup
$P$ from the structure of $pP$ and the known order of $P$.
In particular, we determine the number of cyclic groups of order
$p$ without finding elements that generate them.

We first explain this reconstruction. We then state the existing
basis-construction algorithm that we use to determine the structure
of $pP$. Finally, we show how to obtain a small random list of
generators for $pP$, as required by that algorithm.

\subsection{Recovering a component from its image}

Return to Example~\ref{ex:type}. If $P$ has type $(3,2,1,1)$, then $2P$
has type $(2,1)$. Increasing these surviving exponents by one recovers
$(3,2)$. Since $|P|=2^7$, the missing exponents must sum to
$7-(3+2)=2$. Every missing exponent was one, so there are exactly two
missing factors. This reasoning works for every prime.

\begin{lemma}[Reconstruction and rank]\label{lem:reconstruct}
Let $|P|=p^a$, and suppose the type of $pP$ is
$(\mu_1,\ldots,\mu_s)$, with each $\mu_i\geq1$. Then
\[
 P\cong\left(\bigoplus_{i=1}^s\Z_{p^{\mu_i+1}}\right)
       \oplus\Z_p^{\,c},
\]
where the number of additional order-$p$ factors is
\begin{equation}\label{eq:reconstruct}
 c=a-\sum_{i=1}^s(\mu_i+1).
\end{equation}
Moreover, $s=\rank(pP)\leq\lfloor a/2\rfloor$.
\end{lemma}
\begin{proof}
By Theorem~\ref{thm:structure}, $P$ is a direct sum of cyclic groups of
orders $p^\lambda$. Multiplication by $p$ acts separately on each factor.
If $\lambda\geq2$, its image is cyclic of order $p^{\lambda-1}$; if
$\lambda=1$, its image is trivial. Thus every nontrivial factor of $pP$
corresponds to exactly one factor of $P$ with exponent $\mu_i+1$.
All remaining factors of $P$ have order $p$.

Orders multiply in a direct sum, so the sum of the exponents of all factors
of $P$ equals $a$. Subtracting the known exponents proves
\eqref{eq:reconstruct}. Finally each of the $s$ surviving factors uses at
least two powers of $p$ in $|P|$. Therefore $2s\leq a$.
\end{proof}
The useful bound is on the \emph{rank} of $pP$, not on its order. For
example, if $P=\Z_{p^a}$, then $|pP|=p^{a-1}$, which can be much larger
than $\sqrt{|P|}$. A generic bound depending only on $\sqrt{|pP|}$ would
therefore be insufficient. We need the dependence on rank in the next
proposition.
\color{black}

\subsection{The existing rank-sensitive basis routine}
Given a basis $b_1,\ldots,b_s$, where $b_i$ has order $m_i$,
every group element $g$ can be written uniquely as
$g=x_1b_1+\cdots+x_sb_s$, $0\leq x_i<m_i$.
The integers $x_1,\ldots,x_s$ are called the {\it coordinates} of $g$
in this basis. Given $g$ and the basis, finding these integers
is called the \emph{generalized discrete-logarithm problem}.

Sutherland's
algorithms~\cite{Sutherland11} use this task to build a basis from generators. We use their
published complexity bounds as an external algorithmic theorem; the
remaining reduction is proved here.

\begin{proposition}[Consequence of Sutherland's algorithms ~\cite{Sutherland11}]\label{prop:sutherland}
Given a prime $p$ and a finite list~$S$ generating an Abelian $p$-group $K$
of rank $s$, a basis and its element orders can be computed deterministically
in
\[
 p^{s/2}\operatorname{poly}(|S|,\log n)
\]
time, if $|K|\leq n$ and group operations take unit time. Neither $|K|$ nor $s$
needs to be given as input.
\end{proposition}

\begin{proof} [Proof. Derivation from Sutherland~\cite{Sutherland11}] 
We use Algorithms 1 and 4 and Propositions 1--2 of that paper, including
the discrete-logarithm consequence in its Corollary 1. For nontrivial $K$,
we have $s\geq1$. Write $|K|=p^b$ and let its exponent be $p^h$.
The discrete-logarithm cost is (see Section~\ref{TDLC} in the Appendix)
\[
 T_{\rm DL}(K)=O\!\left(
 \frac{\log(h+1)}{\log\log(h+2)}\log|K|
 +\frac bs p^{s/2}\right).
\]
The cost of constructing a basis of $K$ from the generating
set $S$, also summarized in equation~(3) of that paper,
satisfies
\[
 T_{\rm B}(S)=O\!\left(
 \log^{2+\epsilon}|K|+(|S|-s+1)T_{\rm DL}(K)\right)
\]
for fixed $\epsilon>0$, counting group operations.

Set
 $\epsilon=1$. Since $1\leq b,h,s\leq L$, the first bound
gives $T_{\rm DL}(K)=O(\log^2n+ p^{s/2}\log n)$. Substitution into the second gives
 $T_{\rm B}(S)
 =O\!\left(\log^3n+(|S|+1)(\log^2n+p^{s/2}\log n)\right)$.
This is bounded by $p^{s/2}\operatorname{poly}(|S|,\log n)$, as
claimed. Element
orders can be found by repeatedly multiplying by $p$ until the identity
is reached. 
\end{proof}
\color{black}
\subsection{Obtaining generators by uniform sampling}
Let $G_p$ denote the subgroup of $G$ consisting of elements
whose orders are powers of $p$. To apply the basis-construction
algorithm to $pG_p$, we need a generating list for this subgroup.
We obtain a small such list from random samples of $G$, as
described below, without listing all elements of $pG_p$.

\begin{lemma}[Random generation]\label{lem:generation}
If $K$ is a finite Abelian $p$-group of rank $s$ and $k\geq s$, then $k$
iid uniform samples fail to generate $K$ with probability at most
$p^{s-k}$.
\end{lemma}
\begin{proof}
Let $a_1,\ldots,a_k$ be elements of $K$. Each $a_i$ determines
a coset $a_i+pK$ in the quotient group $K/pK$.
We show that if these cosets span the vector space $K/pK$
over $\F_p$, then $a_1,\ldots,a_k$ generate~$K$.
Let $U$ be the subgroup generated by the list. Spanning the quotient means
$U+pK=K$, so $K/U=p(K/U)$. Iterating gives $K/U=p^h(K/U)$ for every
positive integer $h$. Choose $p^h$ annihilating $K$. Then $K/U=0$, which
means $U=K$. The converse is immediate by taking images in the quotient.

Recall from Section~\ref{S22} that $K/pK\cong\F_p^s$.
Take $k$ independent uniform samples from $K$ and consider
their cosets in $K/pK$. Since every coset has the same number
of elements, these cosets correspond to independent uniform
vectors $y_1,\ldots,y_k$ in $\F_p^s$.
Suppose these vectors do not span $\F_p^s$. We show that there
is a nonzero linear map $f:\F_p^s\to\F_p$ such that
$f(y_i)=0$ for every $i$. Choose a basis $v_1,\ldots,v_t$
of their span, where $t<s$, and extend it to a basis
$v_1,\ldots,v_s$ of $\F_p^s$. Define
$f(a_1v_1+\cdots+a_sv_s)=a_{t+1}$.
Every sampled vector is a linear combination of
$v_1,\ldots,v_t$, so $f$ is zero on every sampled vector.
However, $f(v_{t+1})=1$, so $f$ is not the zero map.

There are $p^s-1$ nonzero linear maps from $\F_p^s$ to $\F_p$:
a linear map is determined by its values on the $s$ basis
vectors, and each value has $p$ choices. For any fixed nonzero
map $f$, exactly $p^{s-1}$ of the $p^s$ vectors satisfy $f(y)=0$.
Thus, a uniform vector satisfies this equality with probability
$1/p$. By independence, all $k$ sampled vectors satisfy it with
probability $p^{-k}$. Since failure to span requires at least
one nonzero map to be zero on all samples, the union bound gives
\[
\Prb[y_1,\ldots,y_k\text{ do not span }\F_p^s]
\leq (p^s-1)p^{-k}
\leq p^{s-k}.
\]
If $s=0$, the quotient is the zero vector space, which is
spanned by every list, so the failure probability is zero.
\end{proof}

We also need to sample the image $pG_p$ without knowing a basis of the input.
Write $n=p^a m$, where $p\nmid m$.
Multiplication by $m$ sends every primary component
corresponding to a prime other than $p$ to zero. One further multiplication by $p$ gives the desired image.

\begin{lemma}[Sampling a primary image]\label{lem:sample}
Write $n=\prod_{p\mid n}p^{a_p}$. For uniform $x\in G$, the element
 $y=({n}/{p^{a_p-1}})x$
is uniform in $pG_p$ and can be computed with $O(\log n)$ table queries.
\end{lemma}
\begin{proof}
Set $m_p=n/p^{a_p}$, so that the displayed scalar is $pm_p$.
By Theorem~\ref{thm:structure}, multiplication by $m_p$ kills every primary
component other than $G_p$: the order of each such component divides $m_p$.
On $G_p$, multiplication by $m_p$ is invertible because
$\gcd(m_p,p^{a_p})=1$. Explicitly, an integer $u$ with
$um_p\equiv1\pmod{p^{a_p}}$ gives its inverse. Following this map by
multiplication by $p$ therefore maps $G$ onto $pG_p$.

For any homomorphism, the preimages of a point in its image form a coset
of its kernel. All these preimage sets have equal size, so the image of a
uniform element is uniform. Finally, the scalar $pm_p$ is at most $n$.
Repeated doubling and addition, using its binary expansion, computes its
multiple of $x$ in $O(\log n)$ table queries.
\end{proof}

For example, in $G=\Z_8\oplus\Z_3$, the order is $24$, and the scalar for
$p=2$ is $24/2^2=6$. It kills the $\Z_3$ component and maps $\Z_8$ uniformly
onto its even residues $2\Z_8$. No element of either component needs to be
identified in advance.
\color{black}
\subsection{The algorithm}
We now have all three ingredients: Lemma~\ref{lem:sample} produces uniform
samples, Lemma~\ref{lem:generation} makes a short list sufficient, and
Proposition~\ref{prop:sutherland} processes that list. Lemma~\ref{lem:reconstruct}
then recovers the missing factors using only integers. See Algorithm~\ref{alg:isomorphism-type}.

If $n=1$, return the trivial type. Otherwise factor $n$ and let $t$ be the
number of distinct prime divisors. We can factor $n$ using the randomized algorithm of Lenstra and
Pomerance~\cite{LenstraPomerance92}, whose expected running time
is $n^{o(1)}$; stopping it after twelve times an upper bound
on this expectation gives a running time of $o(n^{1/4})$
with failure probability at most $1/12$. 

For each $p\mid n$, with $p^{a_p}$ the
exact power of $p$ dividing $n$, perform the following steps.
\begin{enumerate}
\item Set
\[
 k_p=\lfloor a_p/2\rfloor+\lceil\log t\rceil+3.
\]
Draw $k_p$ independent uniform labels $x\in G$. For each one, compute
$y=(n/p^{a_p-1})x$ as in Lemma~\ref{lem:sample}. Let $S_p$ be this list
of image elements.
\item Apply Proposition~\ref{prop:sutherland} to $S_p$. It computes a basis
of the subgroup $K_p=\langle S_p\rangle$ and the orders
$p^{\mu_1},\ldots,p^{\mu_s}$ of its nonidentity basis elements.
\item Form the list $\mu_1+1,\ldots,\mu_s+1$, and append
\[
 c_p=a_p-\sum_{i=1}^s(\mu_i+1)
\]
copies of $1$. Sort this list in nonincreasing order and use it as the
$p$-primary type. If $c_p<0$, report failure.
\end{enumerate}
Return the collection of lists, one for each prime. For an empty basis,
the sum in the last step is zero. An explicit failure or an incorrect list
counts as an unsuccessful run. 

\begin{algorithm}[H]
\caption{Computing the isomorphism type of a finite Abelian group}
\label{alg:isomorphism-type}
\begin{algorithmic}[1]
\Require A finite Abelian group $G$ of known order $n>1$,
with access to group addition and uniform sampling.
\Ensure A collection of primary exponent lists, or \textsc{Failure}.

\State Factor $n=\prod_{p\in\mathcal P}p^{a_p}$,
where $\mathcal P$ is the set of distinct prime divisors of $n$.
\State $t\gets|\mathcal P|$
\State $\mathcal T\gets\varnothing$

\For{each $p\in\mathcal P$}
    \State $k_p\gets
    \lfloor a_p/2\rfloor+\lceil\log_2 t\rceil+3$
    \State $S_p\gets$ an empty list

    \For{$j=1,\ldots,k_p$}
        \State Choose a uniform random label $x\in G$,
        independently of previous choices.
        \State $y\gets(n/p^{a_p-1})x$
        \Comment{Lemma~\ref{lem:sample}}
        \State Append $y$ to $S_p$.
    \EndFor

    \State Use Proposition~\ref{prop:sutherland}
    to compute a basis of $\langle S_p\rangle$.
    \State Let $p^{\mu_1},\ldots,p^{\mu_s}$
    be the orders of its nonidentity basis elements.

    \State $c_p\gets
    a_p-\sum_{i=1}^{s}(\mu_i+1)$
    \State $L_p\gets[\mu_1+1,\ldots,\mu_s+1]$
    \State Append $c_p$ copies of $1$ to $L_p$.
    \State Sort $L_p$ in nonincreasing order.
    \State Insert $(p,L_p)$ into $\mathcal T$.
\EndFor

\State \Return $\mathcal T$
\end{algorithmic}
\end{algorithm}

 In the Algorithm, for an empty basis, we take $s=0$ and the sum defining $c_p$
to be zero. Each returned list $L_p$ specifies the exponents
of the cyclic groups in the proposed $p$-primary decomposition.

\color{black}
\subsection{Correctness and running time}
\begin{proof}[Proof of Theorem~\ref{thm:iso}]
We first show that all sampled subgroups equal their intended images with
high probability. We then bound the cost on every run, including runs on
which some image is not generated.

\paragraph{Correctness.}
Put $s_p=\rank(pG_p)$. By Lemma~\ref{lem:reconstruct},
$s_p\leq\lfloor a_p/2\rfloor$. Lemmas~\ref{lem:sample}
and~\ref{lem:generation} give
\[
 \Prb[K_p\neq pG_p]\leq p^{s_p-k_p}
 \leq2^{-\lceil\log_2 t\rceil+3}\leq\frac1{8t}.
\]
A union bound over the $t$ primes shows that all images are generated with
probability at least $5/6$. On this event, the basis orders returned by
Proposition~\ref{prop:sutherland} describe each whole image. By
Lemma~\ref{lem:reconstruct}, and specifically~\eqref{eq:reconstruct}, the
reconstructed primary types are correct. By Theorem~\ref{thm:structure},
their collection is the type of $G$.

For two groups of equal order, run this procedure on each and compare the
lists. The chance that either type is incorrect is at most $1/3$, again
by a union bound. When both are correct, their comparison decides
isomorphism by Theorem~\ref{thm:structure}.

\paragraph{Cost of the group computations.}
Since $\sum_{p\mid n}a_p\leq\log n$ and $t\leq\log n$, the total
number of samples is
$O(\log n+t\log(t+1))$, which is polynomial in $\log n$.
Lemma~\ref{lem:sample} bounds the cost of forming each image sample.
Even on an unsuccessful run, $K_p$ is a subgroup of $pG_p$.
Equation~\eqref{eq:subgroup-rank} therefore gives
$\rank(K_p)\leq s_p$. By Proposition~\ref{prop:sutherland}, the total
basis-subroutine cost is at most
\[
 \widetilde O\!\left(\sum_{p\mid n}p^{s_p/2}\right)
 \leq\widetilde O\!\left(\sum_{p\mid n}p^{a_p/4}\right).
\]
Every summand is at most $n^{1/4}$, and there are at most $\log_2n$ of
them. Thus this cost is $\widetilde O(n^{1/4})$. When $s_p=0$, its summand
is one. Reconstruction and sorting of the integer lists take polynomial
logarithmic time.

\paragraph{Factoring and elementary oracle operations.} We can factor $n$ using the randomized algorithm of Lenstra and
Pomerance~\cite{LenstraPomerance92} that runs in time $n^{o(1)}$ or by
Costa and Harvey~\cite{CostaHarvey12} deterministic
algorithm that runs in time $\widetilde O(n^{1/4})$,
which is sufficient here for factoring the known order $n$. If the identity is not supplied, compute $nx$ for
any input label $x$, using the fact that every element order divides $n$.
If an imported routine needs an inverse, compute $-x=(n-1)x$ by repeated
doubling. These simulations add only logarithmic factors to the cost.
Together, the three cost estimates prove Theorem~\ref{thm:iso}.
\end{proof}
\color{black}
\section{Decision lower bounds and explicit isomorphisms}\label{sec:consequences}

\subsection{The matching fourth-root decision lower bound}
Theorem~\ref{thm:main} cannot give a decision lower bound on elementary
Abelian $2$-groups: every promised input of order $2^r$ has the same type.
To make decision hard, we need two different types of the same order.
We use the pair from Bshouty~\cite{Bshouty25}. This result is included for
completeness and to make its relation to the new upper bound explicit.

The groups will be products of factors of orders four and two. They satisfy
the same formal identity $4x=0$, but their images under doubling have
different sizes. For order $16$, the pair is $\Z_4^2$ and
$\Z_4\oplus\Z_2^2$; their doubled images have sizes four and two. Detecting
the difference requires an unexpected equality among expressions.

\begin{proposition}[Bshouty's decision lower bound]\label{prop:iso-lower}
For $n=4^m$, let the hidden group be chosen with equal prior probability
from
\[
 H_0=(\Z_4)^m
 \quad\text{and}\quad
 H_1=(\Z_4)^{m-1}\oplus(\Z_2)^2.
\]
Conditional on this choice, give the algorithm a uniform random relabeling
with a fixed identity label. Distinguishing the two cases with average
success probability at least $2/3$ requires $\Omega(n^{1/4})$ queries.
\end{proposition}
\begin{proof}
We use the same symbolic simulation for both groups and compare
each actual execution with this simulation. The simulation does
not depend on which group was chosen. Thus, any advantage in
distinguishing the groups must come from an unexpected equality
between the evaluations of distinct formal expressions.

\paragraph{A common symbolic simulation.}
First fix the algorithm's random coins. Use the symbolic
simulation from Section~\ref{sec:coupling}, with coefficients
in $\Z_4$. The two groups have the same order, and multiplying
any element of either group by four gives zero. The simulation
therefore has the same distribution in both cases: its answers
depend only on formal sums and random choices of labels,
without using the values assigned to the formal symbols.

A decision algorithm outputs a bit rather than a group label.
Thus, by~\eqref{eq:count}, an algorithm performing at most $q$
group additions registers at most $M=3q+1$ distinct formal
expressions, including the expression $0$ for the identity.
Assume first that $M\leq n$.

\paragraph{Collision probabilities.}
Fix the complete symbolic transcript. The formal expressions
are now fixed, while the elements assigned to the formal
symbols remain independent and uniform in the chosen group.

Consider the difference of two distinct formal expressions.
If some coefficient is odd, multiplication by that coefficient
is a bijection on $H_j$. The corresponding multiple of an
independent uniform element is therefore uniform in $H_j$.
Adding the other terms preserves uniformity, so the two
expressions have the same evaluation with probability $1/n$.

Otherwise, all coefficients are even, and at least one
coefficient is two. Doubling a uniform element of $H_j$ gives
a uniform element of $2H_j$, because every element of $2H_j$
has the same number of preimages under doubling. The other
terms also belong to $2H_j$. Their sum with the independent
uniform term is therefore uniform in $2H_j$, so the collision
probability is $1/|2H_j|$. Here,
 $|2H_0|=2^m$,
 and $|2H_1|=2^{m-1}$.
In either coefficient case, a pair collides with probability
at most $2^{-m}$ in $H_0$ and at most $2^{-(m-1)}$ in $H_1$.

There are at most $\binom M2$ pairs of distinct registered
expressions. Applying the union bound and then averaging over
the symbolic transcript, as in Lemma~\ref{lem:collision},
bounds the probability of $\Bad$ in each group by
\[
 \epsilon_0:=\binom M2\,2^{-m},
 \qquad
 \epsilon_1:=\binom M2\,2^{-(m-1)}.
\]

\paragraph{Success when the two groups are equally likely.}
Choose $j\in\{0,1\}$ uniformly. Since the symbolic simulation
is independent of $j$, the bit it outputs equals $j$ with
probability exactly $1/2$.

For each fixed $j$, Lemma~\ref{lem:coupling} compares the actual
execution on a uniformly random labeling of $H_j$ with the
symbolic simulation. Outside $\Bad$, their outputs agree.
Therefore, if the actual execution outputs $j$, either the
simulation also outputs $j$ or $\Bad$ occurs. The actual
success probability can consequently exceed the simulated
success probability by at most $\epsilon_j$.

Averaging over the two choices of $j$ gives
\[
 \Prb[\text{correct decision}]
 \leq \frac12+\frac{\epsilon_0+\epsilon_1}{2}
 = \frac12+\frac32
   \frac{\binom{3q+1}{2}}{\sqrt n},
\]
where we used $n=4^m$. This bound holds for every fixed choice
of the algorithm's random coins. Averaging over those coins
therefore gives the same bound for randomized algorithms.

An algorithm that succeeds with probability at least $2/3$
on every input also succeeds with probability at least $2/3$
under this random choice of group and labeling. Hence,
\[
 \binom{3q+1}{2}\geq \frac{\sqrt n}{9},
\]
which implies $q=\Omega(n^{1/4})$.

Finally, if $M>n$, then
 $q>{(n-1)}/{3}=\Omega(n)$,
which is stronger than the required lower bound.
\end{proof}

The two groups in Proposition~\ref{prop:iso-lower} are not isomorphic:
an isomorphism would map their doubled images bijectively, but those images
have different sizes. Hence distinguishing them is a special case of
isomorphism decision.

\begin{corollary}\label{cor:iso-tight}
Two finite Abelian groups of known order $n$ can be tested
for isomorphism with probability at least $2/3$ in
$\widetilde O(n^{1/4})$ time. Every randomized algorithm
with this success guarantee requires $\Omega(n^{1/4})$
group additions in the worst case along the orders $n=4^m$.
Thus, the time bound is optimal up to factors polynomial
in $\log n$.
\end{corollary}
\begin{proof}
The upper bound is Theorem~\ref{thm:iso}. For the lower bound, compare the
unknown group in Proposition~\ref{prop:iso-lower} with a canonical copy of
$H_0$, represented by coordinates modulo four. Its table answers can be
computed without querying the unknown input. An isomorphism tester would
therefore solve the distinction in Proposition~\ref{prop:iso-lower} with
no greater number of unknown-table queries.
\end{proof}

\subsection{What follows for explicit isomorphisms}

Deciding whether two groups are isomorphic does not require
producing a map between their elements. We now consider an
algorithm that constructs an isomorphism
$f:\F_2^r\longrightarrow G$. For each $i$, let $e_i$ denote
the vector with a $1$ in position $i$ and zeros elsewhere.
The elements $f(e_1),\ldots,f(e_r)$ form a basis of $G$.
Therefore, constructing the map and obtaining the actual
labels of these images together require enough group
additions to construct a basis.

\begin{corollary}[Constructing and evaluating an isomorphism]
\label{cor:explicit-iso}
Suppose a randomized algorithm, given a group
$G\cong\F_2^r$, outputs a description of an isomorphism
$f:\F_2^r\longrightarrow G$ with probability at least $2/3$
on every relabeling. Count all group additions used to
construct this description and to obtain the actual labels
of $f(e_1),\ldots,f(e_r)$. Their total worst-case number is
$\Omega(\sqrt n)$, where $n=2^r$.

If these labels can be obtained from the description without
further group additions, then constructing the description
alone requires $\Omega(\sqrt n)$ group additions in the
worst case.
\end{corollary}

\begin{proof}
Run the algorithm and evaluate its output on
$e_1,\ldots,e_r$. Whenever the algorithm produces a correct
isomorphism, the resulting labels
$f(e_1),\ldots,f(e_r)$ form a basis of $G$. This procedure
therefore constructs a basis with probability at least
$2/3$ on every relabeling. Theorem~\ref{thm:main} implies
that it requires $\Omega(\sqrt n)$ group additions in the
worst case.
\end{proof}

If an algorithm instead supplies a map
$g:G\longrightarrow\F_2^r$, obtaining a basis by this
argument requires finding elements $x_1,\ldots,x_r\in G$
such that $g(x_i)=e_i$. These elements form a basis, but
finding them may require additional group additions.
Those additions must also be counted. Thus, the lower
bound applies to constructing the map and obtaining the
basis labels together; it does not necessarily apply
to constructing the description of the map alone.
\color{black}

\bibliographystyle{plain}
\bibliography{reference}
\newpage
\appendix
\section{The Discrete-Logarithm Cost}\label{TDLC}
Write $|K|=p^b$ and let its exponent be $p^h$.
The discrete-logarithm cost is
\[
 T_{\rm DL}(K)=O\!\left(
 \frac{\log(h+1)}{\log\log(h+2)}\log|K|
 +\frac bs p^{s/2}\right).
\]
Here we explain why the running time includes the term $p^{s/2}$. The approach uses a meet-in-the-middle search: as in the two-sum
problem, they construct two lists and seek one element from each
whose sum equals the target element.
First, suppose that every basis element $b_1,\ldots,b_s$ has
order $p$. Given an element $g$, we want to find coefficients
$x_1,\ldots,x_s\in\{0,\ldots,p-1\}$ such that
$g=x_1b_1+\cdots+x_sb_s$.
There are $p^s$ possible coefficient tuples.
Suppose first that $s=2t$ is even. Form a list of all elements
$x_1b_1+\cdots+x_tb_t$
and a second list of all elements
$g-(x_{t+1}b_{t+1}+\cdots+x_sb_s)$,
letting the coefficients in each list range independently over
$\{0,\ldots,p-1\}$. Store the coefficients together with each
element. Each list has $p^t=p^{s/2}$ entries. An element appearing
in both lists gives
\[
x_1b_1+\cdots+x_tb_t
=
g-(x_{t+1}b_{t+1}+\cdots+x_sb_s),
\]
so the stored coefficients give the desired expression for $g$.
Sorting the lists by element labels allows us to find a match
in $\widetilde O(p^{s/2})$ time. This is the
\emph{baby-step giant-step} method.

When $s$ is odd, we also split one coefficient into two parts.
Specifically, with $m=\lceil\sqrt p\rceil$, every coefficient
$x\in\{0,\ldots,p-1\}$ can be written as $x=u+mv$, where
$0\leq u<m$,
$0\leq v<\lceil p/m\rceil$.
Assigning one part to each list keeps both list sizes
$O(p^{s/2})$.

Suppose $b$ has order $p^e$ and we want to find the coefficient
$x$ in $g=xb$. We determine $x$ one base-$p$ digit at a time.
Multiplying by $p^{e-1}$ gives
$p^{e-1}g=x_0(p^{e-1}b)$, where $x_0=x\bmod p$ is the first
digit. Since $p^{e-1}b$ has order $p$, we can find $x_0$ by the
meet-in-the-middle method above. We then subtract $x_0b$ from
$g$ and use the next lower power of $p$ to recover the next
digit. Repeating this process determines $x$. For several basis
elements, the same idea reduces the computation to problems
in which all basis elements have order $p$. 

Now to construct a basis from a generating set we do the following: For a finite Abelian $p$-group~$K$, Sutherland's
algorithm~\cite[Algorithm~4]{Sutherland11} constructs a basis
from the generating set $S$ one element at a time. It first
chooses an element of largest order as the first basis element.
Suppose the basis constructed so far is $b_1,\ldots,b_t$.
For each remaining generator $g$, a modified discrete-logarithm
computation finds coefficients $c_i$ and replaces $g$ by
$h=g-\sum_{i=1}^{t}c_i b_i$. These coefficients are chosen
so that $h$ is either zero or independent of the current basis:
the cyclic subgroup generated by $h$ intersects the subgroup
generated by $b_1,\ldots,b_t$ only at zero. The algorithm
discards zero elements, adds a remaining element of largest
order to the basis, and repeats these reductions. Subtracting
combinations of the chosen basis elements preserves the subgroup
generated by the entire list. Therefore, when all remaining
elements have become zero, the chosen elements are independent
and generate $K$, so they form a basis.
\end{document}